\documentclass[sn-mathphys]{sn-jnl}

\jyear{2023}%

\theoremstyle{thmstyleone}%
\newtheorem{theorem}{Theorem}
\newtheorem{lemma}[theorem]{Lemma}
\newtheorem{corollary}[theorem]{Corollary}
\usepackage{amsmath}
\usepackage{amsfonts}
\usepackage{amssymb}
\theoremstyle{thmstyletwo}%

\theoremstyle{thmstylethree}%

\begin{document}

\title[Edge-Vertex Domination in UDG]{Complexity and Approximability of Edge-Vertex Domination in UDG}

\author[1]{\fnm{Vishwanath Reddy} \sur{Singireddy}}\email{p20190420@hyderabad.bits-pilani.ac.in}

\author*[1]{\fnm{Manjanna} \sur{Basappa}\footnote{Partially supported by the Science and Engineering Research Board (SERB), Govt. of India, under Sanction Order No. TAR/2022/000397.}} 
\email{manjanna@hyderabad.bits-pilani.ac.in}


\affil[1]{\orgdiv{CS \& IS Department}, \orgname{BITS Pilani}, \orgaddress{\street{Hyderabad Campus}, \city{Hyderabad}, \postcode{500078}, \state{Telangana}, \country{India}}}





\abstract{Given an undirected graph $G=(V,E)$, a vertex $v\in V$ is edge-vertex (ev) dominated by an edge $e\in E$ if $v$ is either incident to $e$ or incident to an adjacent edge of $e$. A set $S^{ev}\subseteq E$ is an edge-vertex dominating set (referred to as \textit{ev}-dominating set and in short as \textit{EVDS}) of $G$ if every vertex of $G$ is \textit{ev}-dominated by at least one edge of $S^{ev}$. The minimum cardinality of an \textit{ev}-dominating set is the \textit{ev}-domination number. The edge-vertex dominating set problem is to find a minimum \textit{ev}-domination number. 
In this paper we prove that the \textit{ev}-dominating set problem is {\tt NP-hard} on unit disk graphs. We also prove that this problem admits a polynomial-time approximation scheme on unit disk graphs. Finally, we give a simple 5-factor linear-time approximation algorithm. 
}

\keywords{Unit disk graph, Edge-vertex dominating set, Approximation algorithm, Edge-facility center location}



\maketitle

\section{Introduction}\label{intro}
Given an undirected graph $G=(V,E)$, the {\it edge neighborhood} of an edge $e'\in E$ is the set of edges in $E$ which share a common vertex $v \in V$ with $e'$, i.e., the set of all edges which are adjacent to $e'$. The set of these neighbors of $e'$ is represented as the set $N_e(e')= \{ f \in E \mid e'$ and $f$ share a common vertex $v\in V \}$. The {\it closed edge neighborhood} of $e'$ is defined as $N_e[e']=N_e(e')\cup \{ e' \}$. The {\it edge neighborhood} of a set $S\subseteq E$ is $N_e(S)=\bigcup_{e' \in S}N_e(e')$. Similarly, the {\it closed edge neighborhood} of a set $S\subseteq E$ is $N_e[S]=\bigcup_{e' \in S}N_e[e']\cup S$. The {\it edge neighborhood of neighborhood} of $e'$ is $N_e(N_e(e'))=N_e^2(e')$. Similarly, the $r$-th edge neighborhood is $N_e^r(e')=N_e(N_e^{r-1}(e'))$ for an integer $r\geq1$.

Given an undirected graph $G=(V,E)$, a vertex $v\in V$ is \textit{ev (edge-vertex)-dominated} by an edge $e\in E$ if $v$ is incident to $e$ (i.e., an endpoint of $e$) or if $v$ is incident to an adjacent edge of $e$. A set $S^{ev}\subseteq E$ is an \textit{edge-vertex dominating set (EVDS)} (referred to as \textit{ev-dominating set}) of $G$ if every vertex of $G$ is $ev$-dominated by at least one edge of $S^{ev}$ (at least two edges for \textit{double edge-vertex dominating set}). The minimum cardinality of an \textit{ev-dominating set} is the \textit{ev-domination number}, denoted by $\gamma_{ev}(G)$. 
A \textit{paired-dominating set (PDS)} of a graph $G(V,E)$ with no isolated vertices is a dominating set $S^{pr}\subseteq V$ and a sub-graph induced by $S^{pr}$ in $G$ have a perfect matching. The minimum cardinality of a \textit{PDS} of $G$ is symbolized as $\gamma_{pr}(G)$. Note that \textit{EVDS} and \textit{PDS} may be completely different subsets of edges in the same graph, and their cardinalities may always be equal (see Fig. \ref{pds}). In Figure \ref{pds}, the blue colored edges represent the {\it EVDS}. However, it should be noted that this set does not fulfill the criteria to be classified as a {\it PDS}; instead, it could correspond to the set of green edges.
Another similar model, called {\it total domination} in a graph, is defined in terms of only vertices instead of \textit{edge-vertex}. A {\it total dominating set} ({\it TDS}) of a graph $G=(V,E)$ is a {\it dominating set} $S^d\subseteq V$ such that every vertex $v\in S^d$ is adjacent to a vertex. The {\it TDS} problem is to find such a set $S^d$ of minimum cardinality. The cardinality of minimum {\it TDS} is denoted by $\gamma_t$. We can also view a {\it TDS} in a graph as a minimum cardinality set of pairs of adjacent vertices (hence, as a set of edges induced on these vertices), where these pairs may share a common vertex. Therefore, a {\it TDS} is also an {\it EVDS} and vice versa. However, a minimum cardinality {\it EVDS} may not be a minimum cardinality {\it TDS} i.e., the set of all vertices incident to edges of a minimum cardinality {\it EVDS} may not necessarily form a minimum cardinality {\it TDS}, and vice versa (see Fig. \ref{tv}). In Figure \ref{tv}, the blue edges represent the minimum cardinality {\it EVDS}, while the red vertices represent the minimum cardinality {\it TDS}. Here we can observe that the cardinality of the vertices incident to blue edges exceeds the minimum cardinality {\it TDS} and the cardinality of the edges incident to the red vertices also exceeds the minimum cardinality {\it EVDS}. \\
\begin{figure}[!htb]
\centering
\includegraphics[scale=.45]{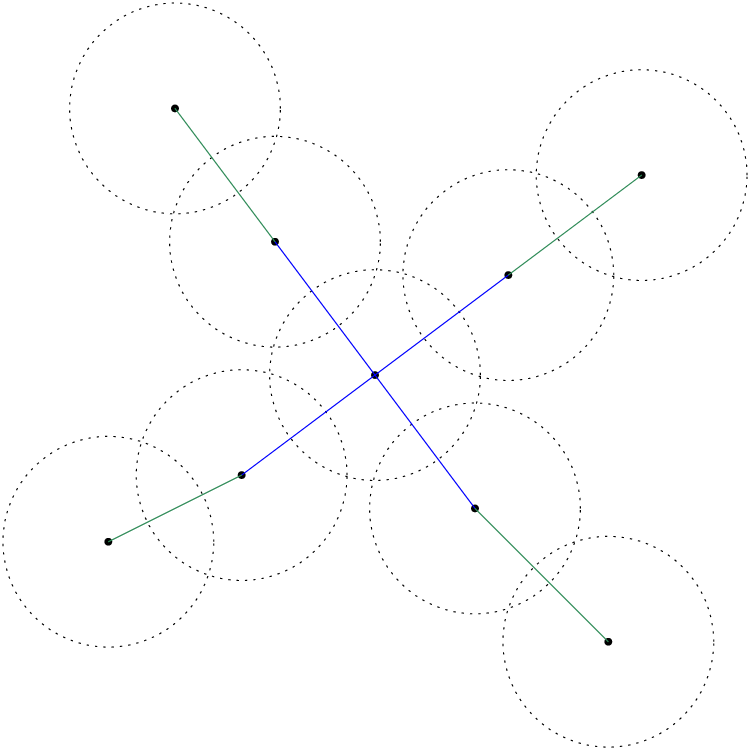}
\caption{{\it EVDS} vs {\it PDS} in {\it UDG}} \label{pds}
\centering
\end{figure}

As well known, a given set $D$ of $n$ disks of unit diameter (hence called \textit{unit disks}) induces a graph $G$, where the graph $G$ is called a unit disk graph ({\it UDG}) $G=(V,E)$ and is an undirected graph such that (i) each vertex $v$ in whose vertex set $V$ corresponds to a disk $d_v\in D$ of unit diameter in the plane, (ii) each edge $(u, v)$ in whose edge set $E$ corresponds to a pair of mutually intersecting disks $d_u$ and $d_v$ in the plane. It is important to note that in {\it UDG} also, a minimum cardinality {\it EVDS} may not be a minimum cardinality {\it TDS} (see Fig. \ref{tv}). Therefore, the study of {\it EVDS} in {\it UDG} holds significant value.

\begin{figure}[!htb]
\centering
\includegraphics[scale=.65]{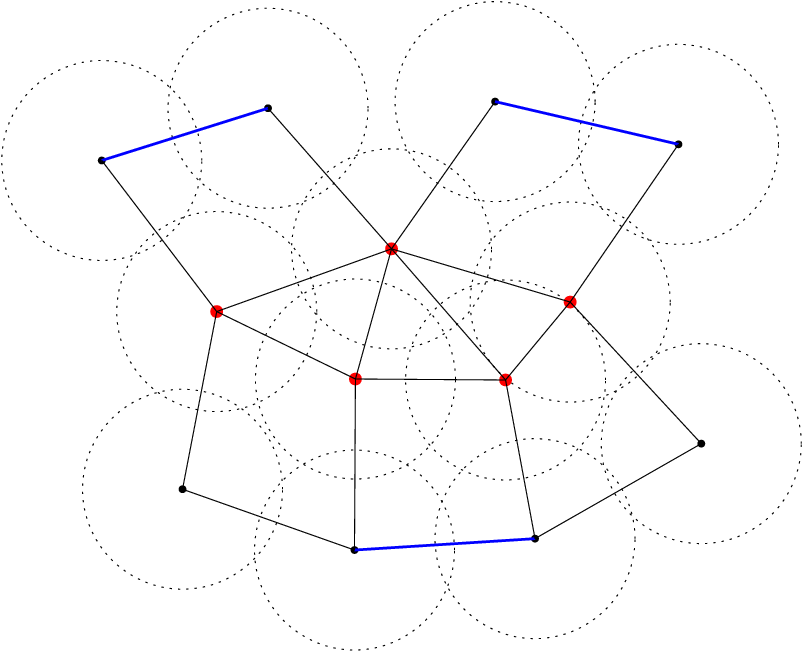}
\caption{{\it EVDS} vs {\it TDS} in {\it UDG}} \label{tv}
\centering
\end{figure}

The edge-vertex dominating set problem may have the following potential applications:
\begin{itemize}
    \item \textit{Edge facility location:} In urban areas, certain facilities, such as parks or street parking zones, that have a significant impact over a wide area. These facilities cannot be adequately represented or modeled by node centers alone because they are not solely accessible from a single entry point. Node centers typically focus on capturing the accessibility of a location from a single point, which may not accurately reflect the spatial distribution and accessibility patterns of these facilities. However, edge centers offer a more faithful representation of such facilities as they consider the multiple entry points and capture the spatial extent and accessibility dynamics more comprehensively. By utilizing edge centers, we can better understand and model the true nature of these impactful urban facilities and their influence on the surrounding area.

    \item \textit{Computer networks:} An {\it EVDS} plays an important role in identifying vulnerable edges in a network. These critical edges would give an attacker control over all connected nodes if compromised. {\it EVDS} helps prioritize security efforts by identifying the most critical edges. Security measures, like encryption or redundancy mechanisms, can then be directed toward securing these connections. Protecting these vital communication links ensures network integrity and minimizes the potential for data manipulation or interception by attackers.
\end{itemize}

\section{Related Work}
The {\it edge-vertex dominating set} and {\it vertex-edge dominating set} in a graph was introduced by Peters \cite{pete1987}. The edge-vertex dominating set and vertex-edge dominating set problems are {\tt NP-complete}, even when restricted to bipartite graphs \cite{lewi2007}. For every nontrivial tree ${\cal T}$, an upper bound on $\gamma_{ev}({\cal T})$ is $(\gamma_{t}({\cal T})+s-1)/2$ where $s$ is the number of support vertices (the vertex adjacent to a leaf) \cite{venk2018}. The \textit{total domination number} ($\gamma_t$) of a tree is equal to the \textit{ev-domination number} ($\gamma_{ev}$) plus one \cite{kris2016}. The vertex-edge dominating set problem in {\it UDG} is {\tt NP-complete} \cite{jena2020}. Also, in \cite{jena2020}, a polynomial time approximation scheme ({\tt PTAS}) is proposed. Finding $\gamma_{ve}$ even in cubic planar graphs is {\tt NP-hard} \cite{ziem2020}. The vertex-edge domination problem can be solved in linear time on block graphs \cite{paul2019}. In the same paper, it is also shown that finding $\gamma_{ve}$ in undirected path graphs is {\tt NP-complete}. Given a connected graph $G$ with $n$ vertices where $n\ge 6$, then we have $\gamma_{ve}(G)\leq \lfloor \frac{n}{3} \rfloor$\cite{zyli2019}. Boutrig et al. \cite{bout2016} gave an upper bound for the independent ve-domination number in terms of the ve-domination number for connected $K_{1,k}$-free graph with $k\geq 3$ and also gave an upper bound on the ve-domination number for connected $C_5$-free graph.  

The double vertex-edge domination was introduced by Krishnakumari et al. \cite{kris2017}. They showed that finding $\gamma_{dve}$ in a bipartite graph is {\tt NP-complete} and also proved that for every non-trivial connected graph $G$, $\gamma_{dve}(G)\geq \gamma_{ve}(G)+1$, and $\gamma_{dve}({\cal T})=\gamma_{ve}({\cal T})+1$ or $\gamma_{dve}({\cal T})=\gamma_{ve}({\cal T})+2$ for any tree ${\cal T}$. Finding $\gamma_{dve}$ in chordal graphs is {\tt NP-complete} \cite{venk2019}. They gave a linear time algorithm to find $\gamma_{dve}$ in proper interval graphs and also showed that finding $\gamma_{dve}$ in general graphs with vertices having a degree at most 5 is {\tt APX-complete}. The double version of edge-vertex domination was studied by Sahin and Sahin \cite{sahi2021}. They also gave the relationship between $\gamma_{dev}$ and $\gamma_{dve}$, $\gamma_{t}$, $\gamma_{ev}$ for trees and graphs, and also gave formulas to determine the \textit{double ev-domination number} of paths and cycles. Sahin and Sahin \cite{SAHI2020} proved that the total ev-dominating set problem is {\tt NP-hard} for bipartite graphs. They also showed that $(n-l+2s-1)/2$ is the upper bound for $\gamma^t_{ev}$ for a tree ${\cal T}$ with order $n$, $l$ leaves and $s$ supporting vertices. To the best of our knowledge, in the literature, the ev-dominating set problem is not yet studied in the context of geometric intersection graphs. 

\subsection{Our Contribution}
In this article, we study the {\it EVDS} problem on unit disk graphs. We show that the decision version of this problem is {\tt NP-complete} in UDGs. We also prove that this problem on UDG admits a polynomial time approximation scheme (PTAS). We finally present a simple 5-factor linear-time approximation algorithm.

\section{Hardness Results}
In this section, we show that the decision version of the {\it EVDS} is {\tt NP-complete}, as stated below. We describe a polynomial time reduction from the vertex cover problem, which is known to be {\tt NP-complete} in planar graphs with maximum degree 3 \cite{grey1979}, to {\it EVDS} problem on {\it UDG}. \\
\newline
\textbf{The {\it EVDS} problem on UDGs} (EVDS-UDG)\\
\textbf{Instance:} A UDG $G=(V,E)$ and a positive integer $k$.\\
\textbf{Question:} Does there exist an edge-vertex dominating set $S^{ev}$ of $G$ such that $\vert S^{ev}\vert \leq k$.\\

\begin{lemma}(\cite{vali1981})\label{lem1} 
An embedding of a planar graph $G=(V,E)$ with maximum degree $4$ in the plane is possible such that this embedding uses only $O(\lvert V \rvert ^2 )$ area and its vertices are at integer coordinates, and its edges are drawn so that they are along the grid line segments of the form $x = i$ or $y = j$, for some $i$ and $j$, where $i,j \in \mathbb{Z}^+$.
\end{lemma}

Biedl and Kant \cite{Bied1998} gave an algorithm that produces this kind of embedding in linear time (see Fig. \ref{fig1}).

\begin{corollary}(\cite{jena2020})\label{cor1} An embedding of a planar graph $G=(V,E)$ with $\lvert V \rvert \geq3$ and maximum degree 3 in the plane can be constructed in polynomial time, where the embedding is such that the vertices of $G$ are at $(4i,4j)$ and the edges of $G$ are drawn as a sequence of consecutive line segments along the lines $x=4i$ or $y=4j$, for some $i$ and $j$.
\end{corollary} 

\begin{figure}
\includegraphics[scale=0.7]{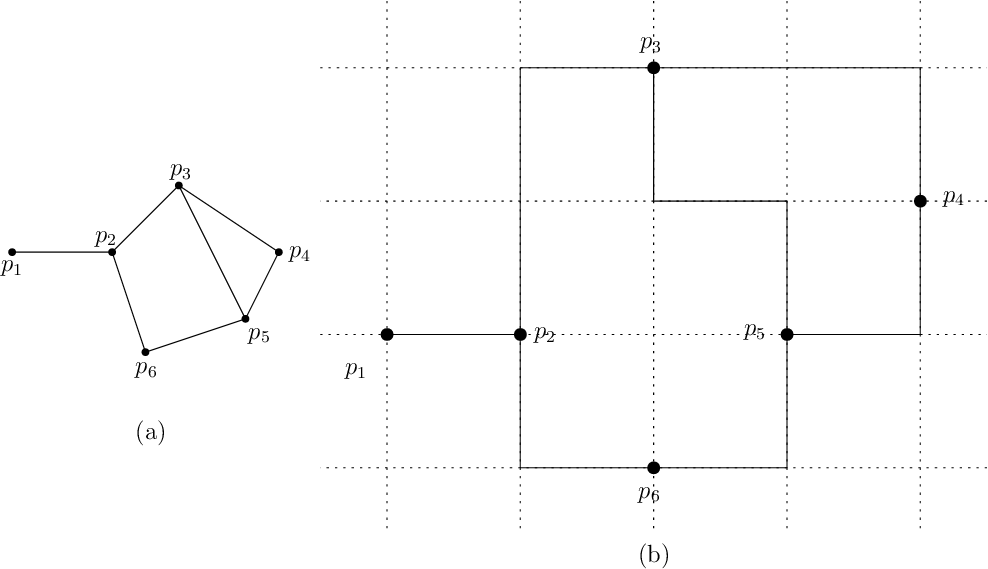}
\centering
\caption{(a) A planar graph $G$ with max degree 3, and (b) The embedding of $G$ on a grid.} \label{fig1}
\centering
\end{figure}

\begin{lemma}\label{lem2}
Let $G=(V,E)$ be an instance of the vertex cover problem with a number of edges at least 2 and a maximum degree 3. An instance $G'=(V',E')$ of {\it EVDS-UDG}  can be constructed from $G$ in polynomial time.
\end{lemma}
\begin{proof}
The construction of $G'=(V',E')$ from $G=(V,E)$ is as follows:
First, using one of the algorithms discussed in \cite{hopc1974,itai1982} we embed the graph $G=(V, E)$ into a grid of size $4n \times 4n$ such that each edge of $E$ is composed of a sequence of horizontal or vertical line segment(s), each of whose length is four units long. The points $\{p_1,p_2,\dots,p_n\}$ are referred to as the \textit{node points} in the embedding with respect to the vertex set of $G$ (see Fig. \ref{fig1}(a), Fig. \ref{fig1}(b) and the corresponding {\it UDG} $G'$ in Fig. \ref{fig2}). In the embedded graph, for each edge of length greater than four units, we add a joint point to join two line segments in the embedding other than the node points. Name these points as the \textit{joint points} (see empty circles in Fig. \ref{fig2}). Then for each line segment with \textit{joint points} as both of its end points in the embedding, we add three extra points such that each of these extra points is at a distance of 1 unit from its neighbor extra point(s) placed on the same segment, at least 1 unit from the corresponding joint points. Similarly, for each line segment with a \textit{node point} as its endpoint, we add four extra points each at a distance of 0.8 unit from its neighbor extra point(s), also from the end points of the segment on which we are placing the extra points. Name these extra points (from both the above cases) as the \textit{added points} (see filled square points in Fig. \ref{fig2}). 

Let $A$ be the set of added points and $J$ be the set of joint points. We construct a UDG $G'=(V',E')$  where the vertex set $V'=V\cup A\cup J$, and there is an edge between two vertices of $V'$ if and only if the distance between them is at most 1 unit (see Fig. \ref{fig2}). If $l$ is the total number of line segments in the embedding, then $\lvert A \rvert \leq 4l$ and $\lvert J \rvert \leq (l- \lvert E \rvert) $. It follows from Lemma \ref{lem1} that $l$ is at most $O(n^2)$. Clearly, the graph defined by the intersection of unit disks centered at points in $V'$ is a unit disk graph. Since both the sets $\lvert V' \rvert $ and $\lvert E' \rvert $ are bounded by $O(n^2)$, we can construct $G'$ from $G$ in polynomial time. 
\end{proof}

\begin{figure}
\centering
\includegraphics[scale=0.7]{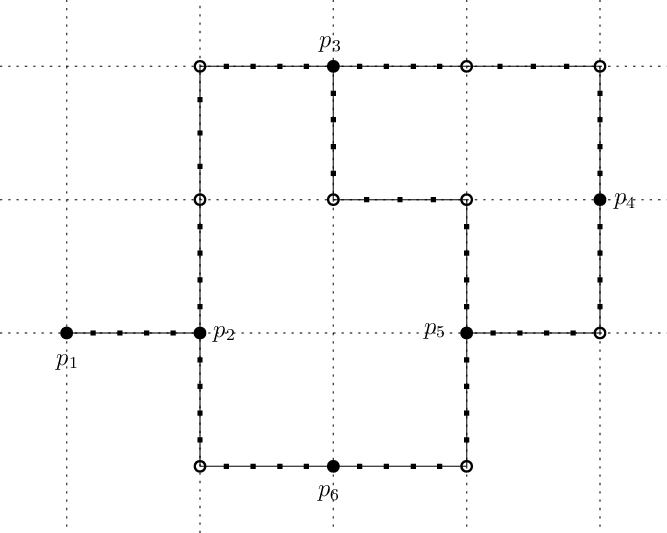}
\caption{A unit disk graph construction from embedding} \label{fig2}
\centering
\end{figure}	

\begin{lemma}\label{lem3}

EVDS-UDG $\in{\tt NP}$.
\end{lemma}
\begin{proof}
Given a subset $S^{ev}\subseteq E$ and a positive integer $k$, we can verify that $S^{ev}$ is an edge-vertex dominating set of size at most $k$ in polynomial time by checking whether each vertex $v\in V$ is {\it ev}-dominated by an edge $e\in S^{ev}$ all in $O(\lvert V'\rvert \lvert E'\rvert)$ time. 
\end{proof}

We prove the {\tt NP-hard}ness of the {\it EVDS-UDG} problem by reducing the decision version of the {\it vertex cover} problem on a planar graph with maximum degree 3 to the {\it EVDS-UDG} problem. Let $G=(V,E)$ be a planar graph with a maximum degree of 3. Then from Lemma \ref{lem2}, we can construct an instance $G'=(V',E')$ of {\it EVDS-UDG} in polynomial time.

\begin{figure}
\centering
\includegraphics[scale=0.7]{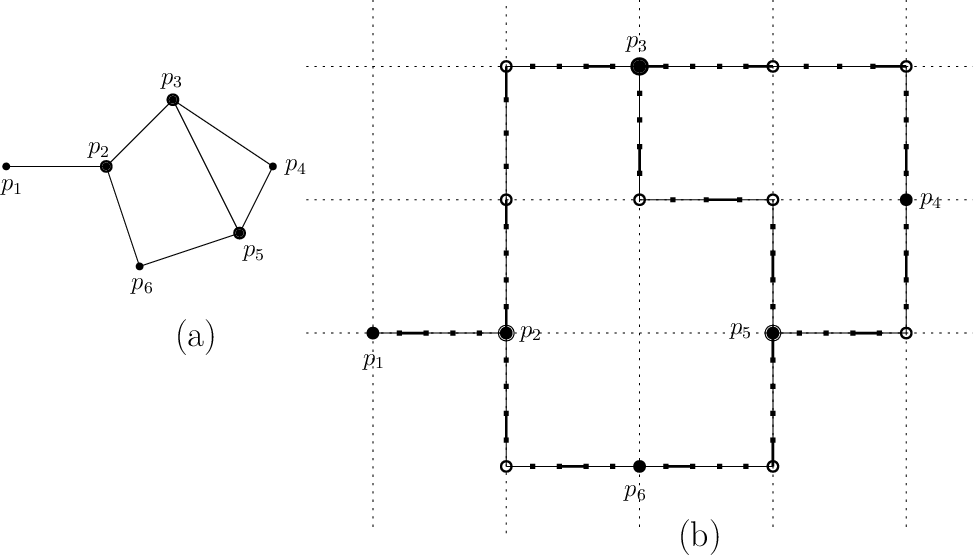}
\caption{(a) Vertex cover of $G$. (b) {\it EVDS} of UDG.} \label{fig3}
\centering
\end{figure}	

\begin{lemma}\label{lem4}
$G$ has a vertex cover of size at most $k$ if and only if $G'$ has an edge-vertex dominating set of size at most $k+l$.
\end{lemma}

\begin{proof}
Let $D\subseteq V$ be a vertex cover of $G$ of cardinality at most $k$. Let $D'$ be the set of vertices all happen to be the \textit{node point}s of $G'$ that correspond to the vertices in $D$.

{\it Necessary}: Now, for every vertex $v$ in $D'$, choose any one edge of $G'$ that is incident to $v$ where the path does not lead to a pendant vertex in $G$. Represent these chosen edges as the set $N'$. Since $\lvert D'\rvert \leq k$, the cardinality of $N'$ is at most $k$. We can see that for any edge $(p_i, p_j)\in E$ of $G$, we have a simple path in $G'$, consisting of at least one line segment and only added points and joint points between $p_i$ and $p_j$. Let us introduce the notation $p_i\rightsquigarrow p_j$ to denote this path for any pair of node points $(p_i, p_j)\in E$, where the node point $p_i$ is a vertex in $D'$. Next, traverse every such path exactly once starting at a vertex $v\in D'$, and initially choose the fourth edge from $(v, u)\in N'$ (not counting $(v,u)$) and then on every fourth edge until reaching $p_j$. We repeat this for every vertex $v\in D'$ except for the paths that will lead to pendant vertices of $G$, which will ensure that every path $p_i\rightsquigarrow p_j$ is traversed exactly once because $D$ is a vertex cover. Similarly, for every path $p_i\rightsquigarrow p_j$, where $p_i\in D$ and $p_j$ is a pendent vertex, choose the second edge of $G'$ from $p_j$ and continue selecting every fourth edge until reaching $p_i$. Let $N''$ be the set of all these chosen edges. Observe that the edges in $N'\cup N''$ are chosen such that there is at least one edge and at most two edges contained in $N'\cup N''$ for each of $l$ segments (see the darkened edges in Fig. \ref{fig3}(b)). Moreover, the way edges in $N'\cup N''$ are chosen ensures that every two added or joint or node points between any two consecutive of these edges are {\it ev}-dominated by them. The cardinality of $N''$ is $l$ since each line segment consists of at most four added points in the embedding. Therefore, $N'\cup N''$ is an {\it ev}-dominating set for $G'$ and $(\lvert N' \rvert +\lvert N'' \rvert) \leq k+l$.


{\it Sufficiency}: To prove the sufficiency, consider any edge-vertex dominating set $S^{ev}\subseteq E'$ for $G'$, of size at most $k+l$. We know that for any node point $p_i\in V'$, the degree of $p_i$ is at most 3 in both $G$ and $G'$. Let $p_i\rightsquigarrow p_{j_t}$ ($t=1,2,3$) be the three paths in $G'$ as defined above, i.e., all the other vertices through which the path $p_i\rightsquigarrow p_{j_t}$ traverses are only the joint points and added points.
Let $\zeta(p_i\rightsquigarrow p_{j_t})$ be the subset of edges of $E'$ that appear in the path $p_i\rightsquigarrow p_{j_t}$. For any node point $p_i$ in $G'$, let $\ell_t= \lvert\zeta(p_i\rightsquigarrow p_{j_t})\cap S^{ev}\rvert$ be the number of edges of that path contained in the {\it EVDS} $S^{ev}$. Let $\jmath_t$ be the number of line segments that constituted the path $p_i\rightsquigarrow p_{j_t}$ in the embedding. Observe that $\ell_t$ is equal to $\jmath_t$ or $\jmath_t+1$ due to the construction of $G'$ from the embedding of $G$. Now, identify a node point $p_i$ in $G'$ such that the degree of $p_i$ is at least 2 and exactly one path $p_i\rightsquigarrow p_{j_1}$ has its $\ell_1$ equal to $\jmath_1+1$ and the remaining paths $p_i\rightsquigarrow p_{j_2}$ (and $p_i\rightsquigarrow p_{j_3}$) have their counts $\ell_2=\jmath_2$ (and $\ell_3=\jmath_3$). Pick this node point $p_i$ into a vertex cover $D$. Remove the part of $G'$ induced by these paths $p_i\rightsquigarrow p_{j_1}$, $p_i\rightsquigarrow p_{j_2}$, and $p_i\rightsquigarrow p_{j_3}$ (however, retain the node points $p_{j_1}$, $p_{j_2}$, and $p_{j_3}$ in the remaining $G'$). This will guarantee that the edges of $G$ corresponding to these paths are covered by the vertex $p_i\in D$. Repeat this procedure on the remaining $G'$. To start with, there must exist at least one such $p_i$; otherwise, the sum $\sum_{p_i\rightsquigarrow p_{j_t}}\ell_t$ over all paths in $G'$ would exceed $k+l$, a contradiction. Hence, $D$ is a vertex cover for $G$ and $\lvert D\rvert\leq k$.

The construction of $S^{ev}$ from $D$ and vice versa both take polynomial time. Thus the lemma follows.
\end{proof}

\begin{theorem}
The EVDS-UDG problem is {\tt NP-complete}.
\end{theorem}
\begin{proof}
Follows from Lemmas \ref{lem3} and \ref{lem4}. 
\end{proof}

\section{Polynomial Time Approximation Scheme}\label{poly}
In this section, we propose a PTAS for the {\it EVDS} set problem in a UDG. It is based on the concept of $m$-separated collection of subsets, which was introduced by Nieberg and Hurink \cite{nieb2006}. This concept was used by many other authors to develop PTAS (for e.g., the Roman dominating set \cite{SHAN2007}, minimum Liar's dominating set \cite{JALL2020}, vertex-edge dominating set \cite{jena2020}). However, we adopted that concept here quite differently from these as we have to select edges to dominate vertices in the {\it EVDS} problem. The {\it EVDS} is a domination model with some relevant applications (as we have discussed in Section \ref{intro}) different from other related models such as {\it VEDS}, {\it TDS}, {\it PDS}. Let $G=(V,E)$ be a {\it UDG}. Let $h(e_1,e_2)$ denote the minimum number of edges in a simple path between the endpoints of the edges $e_1$ and $e_2$. Consider any two subsets $E_1\subseteq E$ and $E_2\subseteq E$, $h(E_1,E_2)$ is defined as the minimum number of edges between any two edges $e_1\in E_1$ and $e_2\in E_2$.
We use $EVD(A)$ to denote an ev-dominating set and $EVD_{opt}(A)$ to denote the optimal ev-dominating set of the edge-induced subgraph corresponding to $A(\subseteq E)$ (i.e., the subgraph induced by the set of edges $A(\subseteq E)$ and the endpoints of edges in $A$).
 
Let $S$ be a set of $k$ pairwise disjoint subsets of $E$, i.e., $S_i\subset E$ for $i=1,2,\dots,k$. If $h(S_i,S_j)\geq m$, for $1\leq i$, $j\leq k$ and $i\neq j$, then $S$ is called as the \textit{$m$-separable collection} of subsets of $E$ (see Fig. \ref{fig4} for $m=4$).

\begin{lemma} \label{lemp1}
In a graph $G=(V,E)$, if $S=\{S_1,S_2,\dots,S_k\}$ is a $4$-separated collection of $k$ subsets of $E$, then \[ \sum_{i=1}^{k}\lvert EVD_{opt}(S_i) \rvert \leq \lvert EVD_{opt}(E) \rvert. \]
\end{lemma}

\begin{proof}
Let $A_i$ be the set of edges that are adjacent to edges of $S_i$ for each $i=1,2,\dots,k$ and $R_i$ the set of edges such that $R_i=S_i\cup A_i$. The edges in sets $R_1,R_2,\dots,R_k$ are pairwise disjoint, since the set $S$ is a 4-separated collection of subsets of edges i.e., $(R_i\cap R_j)=\emptyset$, where $i\neq j$. Hence, the edges of $EVD_{opt}(E)\cap R_i$ will ev-dominate every vertex in $S_i$, since $EVD_{opt}(E)$ will ev-dominate every vertex $v\in V$. 
On the other hand, also $EVD_{opt}(S_i)\subset R_i$ ev-dominates every vertex of $S_i$, with a minimum number of edges of $G$. This implies that $\vert EVD_{opt}(S_i)\vert \leq \vert EVD_{opt}(E)\cap R_i \vert$. For all $k$ subsets of edges in the 4-separated collection $S$, we get
\begin{align*}
\sum_{i=1}^{k}\vert EVD_{opt}(S_i)\vert  \leq \sum_{i=1}^{k}\vert (EVD_{opt}(E)\cap R_i) \vert 
 \leq \vert EVD_{opt}(E)\vert.
\end{align*}
\end{proof}

The above Lemma \ref{lemp1} states that a 4-separated collection of subsets of edges $S$ will give a lower bound on the cardinality of an {\it EVDS}. Hence, we can get an approximation for the {\it EVDS} in $G$, if we are able to enlarge $S_i$ to subsets $Q_i\subset E$, in such a way that {\it EVDS} of expansions are bounded locally and dominate every $v\in V$ globally. 

\begin{figure}[!htb]
\centering
\includegraphics[scale=0.65]{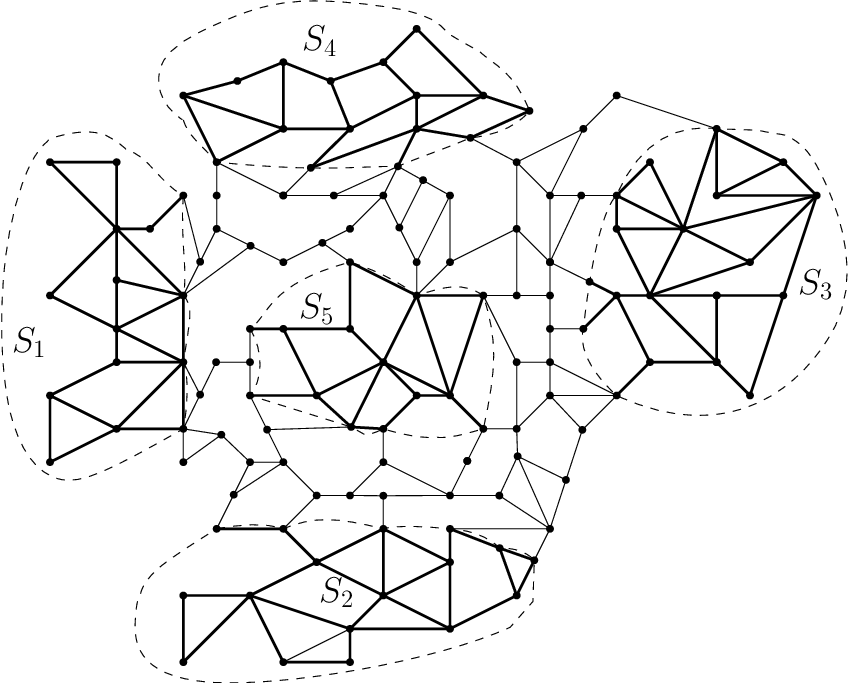}
\caption{4-Separable collection of edge sets $S=\{ S_1,S_2,S_3,S_4,S_5 \}$} \label{fig4}
\centering
\end{figure}

\begin{lemma}\label{lemp2}
In a graph $G=(V,E)$, let $S=\{S_1,S_2,\dots,S_k \}$ be a 4-separated collection of subsets of edges and $Q=\{Q_1,Q_2,\dots,Q_k \}$ be a collection of subsets of $E$ with $S_i\subseteq Q_i$ for every $i=1,2,\dots,k$. If there is  a $\rho\geq 1$ such that \[ \lvert EVD_{opt}(Q_i) \rvert \leq \rho\lvert EVD_{opt}(S_i) \rvert  \] holds for every $i=1,2,\dots, k$, and if $\bigcup_{i=1}^{k} EVD_{opt}(Q_i)$ is an edge-vertex dominating set of $G$, then $\sum_{i=1}^{k}\lvert EVD_{opt}(Q_i) \rvert $ is a $\rho$-approximation of minimum {\it EVDS} set of $G$.
\end{lemma}

\begin{proof}
From Lemma \ref{lemp1} we have,
\[ \sum_{i=1}^{k}\lvert EVD_{opt}(S_i) \rvert \leq \vert EVD_{opt}(E)\vert . \]
Hence, $\sum_{i=1}^{k}\vert EVD_{opt}(Q_i)\vert \leq \rho \sum_{i=1}^{k}\vert EVD_{opt}(S_i)\vert \leq \rho \vert EVD_{opt}(E)\vert  $. 

\end{proof}
In the following section, we discuss a procedure to construct the subsets $Q_i\subset E$, that contains a 4-separated collection of edges $S_i\subset Q_i$, in such a way that a local $(1+\epsilon)$-approximation can be guaranteed. The union of the respective local {\it EVDS} will ev-dominate the entire vertex set of $G$, which results in a global $(1+\epsilon)$-approximation for the {\it EVDS} problem.

\subsection{Subset Construction}

Here, we discuss the construction of the 4-separated collection of subsets of edges, $S=\{S_1,S_2,\dots,S_k \}$ and the respective enlarged subsets $Q=\{Q_1,Q_2,\dots,Q_k\}$ of $E$ such that $S_i\subseteq Q_i$ for every $i=1,2,\dots,k$. The basic idea of the algorithm is as follows. 
We start with an arbitrary edge $e\in E$ and consider the $r$-th edge neighborhood of $e$, for $r=0,1,2,\dots$, with $N_e^0[e]=e$. We compute the {\it EVDS} for these edge neighborhoods until the following condition holds
\begin{equation}\label{eq1}
\vert EVD(N_e^{r+4}[e])\vert > \rho\vert EVD(N_e^r[e])\vert 
\end{equation}
Let $r_1$ be the smallest $r$ that violates the above inequality (\ref{eq1}). Let $S_1=N_e^{r_1}[e]$, $Q_1=N_e^{r_{1}+4}[e]$. Then iteratively, let $S_i=N_e^{r_i}[e]$, $Q_i=N_e^{r_i+4}[e]$, $E_{i+1}=E_i \backslash (N_e^{r_i+4}(e))$ for $i=1, 2, \ldots, k$, where $E_1=E$ and $k$ is such that $E_{k+1}=\emptyset$. We follow this procedure iteratively for each graph induced by $E_{i+1}$ and  until $E_{i+1}=\emptyset$, finally returning the sets $S=\{S_1,S_2,\dots,S_k\}$ and $Q=\{Q_1,Q_2,\dots,Q_k\}$, where $r_2, r_3, \ldots, r_k$ are the smallest values of $r$ violating inequality (\ref{eq1}), corresponding to the 2nd, 3rd, \ldots, $k$th iteration of the above edge-neighborhood growing procedure.

We find the edge-vertex dominating set of the $r$-edge neighborhood $EVD(N_e^r[e])$ of an edge $e$, with respect to the graph $G$ as follows. Find a maximal matching $M$ for the graph induced by the edges of $N_e^r[e]$. We can observe that the edges in $M$ form an edge-vertex dominating set for the graph induced by $N_e^r[e]$. Hence, as the following lemma says, $EVD(N_e^r[e])=M$.

\begin{lemma}\label{lemp3}
A maximal matching $M$ of the graph $G'=(V',E')$ induced by the edges in $N_e^r[e]$, is an {\it EVDS} of $N_e^r[e]$.
\end{lemma}
\begin{proof}
For the contradiction, assume that $M$ is not an {\it EVDS} of the graph $G'=(V',E')$ induced by the edges in $N_e^r[e]$. It means that there exists a vertex $v\in V'$ which is incident to an edge $e'\in E'$ such that $N_e[e']\cap M=\emptyset$. It contradicts that $M$ is a maximal matching in $G'$ as the set $M\cup\{e'\}$ is a matching in $G'$. Thus, the lemma follows. 
\end{proof}

\begin{lemma}\label{lemp4}
If $G'=(V',E')$ is a UDG induced by the edges in $N_e^r[e]$ and $M$ is the maximal matching of $G'$ then $ \vert  EVD(N_{e}^{r}[e]) \vert  \leq O(r^2)$.
\end{lemma}
\begin{proof}
First, we find a maximal matching $M$, before finding the {\it EVDS} in $G'=(V',E')$ which is induced by the edges of $N_e^r[e]$. The number of edges in $M$ of $G'$ is bounded by the number of unit disks that are packed in a disk of radius $r+2$ and centered at the middle of the edge $e$. Hence, $\vert M\vert \leq (r+2)^2$ and the cardinality of $EVD(N_e^r[e])$ is bounded by $\vert M\vert $ (see Lemma \ref{lemp3}). Therefore, we have 
\[ \vert EVD(N_e^r[e])\vert \leq \vert M\vert \leq (r+2)^2\leq O(r^2). \] 
\end{proof}

\begin{theorem}\label{thmp1}
There exists an $r_1$ which violates the following inequality. \[ \vert EVD(N_e^{{r_{1}}+4}[e])\vert > \rho\vert EVD(N_e^{r_{1}}[e])\vert  \] where $\rho=1+\epsilon$ and $r_1$ is bounded by $O(\frac{1}{\epsilon}\log {\frac{1}{\epsilon}})$.
\end{theorem}
\begin{proof}
On contrary, without loss of generality for $r_1$, assume that there exists an edge $e\in E$ such that 
\begin{equation*}\label{eq2}
 \vert EVD(N_e^{r+4}[e])\vert > \rho\vert EVD(N_e^{r}[e])\vert 
\end{equation*}
 for all $r\geq r_1$. Then, from Lemma \ref{lemp4} , we have  \[ (r+6)^2 \geq \vert EVD(N_e^{r+4}[e])\vert . \]
Hence, when $r$ is even we have,

\begin{align}\label{eq3}
(r+6)^2\geq \vert EVD(N_e^{r+4}[e])\vert   > \rho \vert EVD(N_e^{r}[e])\vert >   \cdots  > \rho^{\frac{r}{2}}\vert EVD(N_e^{2}[e])\vert \geq \rho^{\frac{r}{2}} 
\end{align}
and when $r$ is odd, we have,
\begin{align}\label{eq4}
(r+6)^2\geq \vert EVD(N_e^{r+4}[e])\vert   > \rho \vert EVD(N_e^{r}[e])\vert > \cdots  > \rho^{\frac{r-1}{2}}\vert EVD(N_e^{1}[e])\vert \geq \rho^{\frac{r-1}{2}}
\end{align}
Now, we can observe that in both the inequalities (\ref{eq3}), (\ref{eq4}) on the left-hand side we have a polynomial in $r$ which is at least the right-hand side value which is exponential in $r$, it is a contradiction. Therefore, for all $r\geq r_1$ the inequality (\ref{eq3}) cannot hold, hence there exists such $r_1$. Ultimately, $r_1$ depends only on $\rho$, not on the size of the edge-induced subgraph by $N_e^{r+4}[e]$. As in \cite{nieb2006}, we can argue that $r_1$ is bounded by $O(\frac{1}{\epsilon}\log {\frac{1}{\epsilon}})$ where $\rho=1+\epsilon$. 
\end{proof}

\begin{lemma}\label{lemp5}
Given an $\epsilon > 0$, for an edge $e\in E$, $EVD_{opt}(Q_i)$ can be computed in polynomial time.
\end{lemma}
\begin{proof}
From the way of construction of $Q_i$, we can see that $Q_i\subseteq  N_e^{r+4}[e]$. The cardinality of {\it EVDS} of $N_e^{r+4}[e]$ is bounded by $O(r^2)$ (see Lemma \ref{lemp4}), where $r$ is bounded by $O(\frac{1}{\epsilon}\log {\frac{1}{\epsilon}})$ (see Theorem \ref{thmp1}). Hence, we need at most $O(n^{r^2})$ possible combinations of $O(r^2)$-tuples of vertex points to check whether the selected tuple is an {\it EVDS} of $Q_i$.  
\end{proof}

\begin{lemma}\label{lemp6}
$\bigcup_{i=1}^{k} EVD(Q_i)$ is an edge-vertex dominating set in $G=(V,E)$.
\end{lemma}
\begin{proof}
It follows from the construction of the collection of subsets of edges $\{ Q_1, Q_2,$ $\dots, Q_k \}$ that each edge that is incident to a vertex $v\in V$ belongs to a specific subset $Q_i$ and $EVD(Q_i)$ is an {\it EVDS} of the graph induced by the edges of $Q_i$. Therefore, every vertex $v\in V$ is incident to at least one edge $e$ such that at least one edge of $N_e[e]$ is in $\bigcup_{i=1}^{k} EVD(Q_i)$.  
\end{proof}

\begin{corollary}\label{cor}
$\bigcup_{i=1}^{k} EVD_{opt}(Q_i)$ is an edge-vertex dominating set in $G=(V,E)$, for the collection of subsets of edges $Q=\{ Q_1,Q_2,\dots,Q_k\}$.

\end{corollary}

\begin{theorem}\label{thmp2}
For a given unit disk graph and an $\epsilon>0$, there exists a PTAS (an $(1+\epsilon)$-approximation) algorithm for the {\it EVDS} problem with running time $n^{O(c^2)}$, where $c=(\frac{1}{\epsilon}\log{\frac{1}{\epsilon}})$.  
\end{theorem}
\begin{proof}
Follows from Corollary \ref{cor} and Lemma \ref{lemp5}. 
\end{proof}

\section{5-Factor Approximation Algorithm}

 In this section, we present a $5$-factor approximation algorithm for the {\it EVDS} problem on UDG. Let ${\cal S}$ be a set of $n$ points given in the Euclidean plane. We join two of these points with an edge if the distance between those two points is less than or equal to 1 unit. Let $E$ be the set of such edges with cardinality $m$ and $V$ be the set of vertices corresponding to points in ${\cal S}$. The graph induced by $V$ and $E$ will form a UDG since the distance between any two end-points of $e\in E$ is at most 1. Assume that such an UDG has no isolated vertex, otherwise {\it EVDS} does not exist. To present an approximation algorithm, we consider an axis-parallel rectangular region ${\cal R}$ that contains UDG. We then partition the region ${\cal R}$ into grid cells by a hexagonal tessellation, where each hexagonal cell is of side length $\frac{1}{2}$. Hence the maximum distance between any two points inside a cell is at most 1. Assume that no point in $V$ lies on the boundary of any hexagon in the partition. 

\begin{lemma}\label{lema1}
Any edge $e$ with its two endpoints lying in adjacent hexagons can ev-dominate every point in those two hexagons.
\end{lemma}
\begin{proof}
It follows from the fact that there will be an edge between any two points that lie within the same hexagon since the distance between them is at most 1. Therefore, an edge $e\in EVDS$ whose endpoint lies in that hexagon will ev-dominate every other point in that hexagon. 
\end{proof}

  The outline of the algorithm is as follows. Initialize the set $S^{ev}$ (which will hold the edges of {\it EVDS}) initially as empty. Now, arbitrarily pick an edge $e\in E$ whose endpoints lie in different cells. Add this edge to $S^{ev}$ and set $E=E\backslash \{e\}$. Mark all points that are ev-dominated by $e$. If there are any unmarked vertices, now choose an edge $e\in E$ that is incident to any of the unmarked vertices, with its other endpoint lying in a different cell. Add $e$ to $S^{ev}$ and mark all the unmarked points that are ev-dominated by $e$. Repeat this process until every point in $V$ is marked (see Algorithm \ref{alg1}).

%

\begin{algorithm}\label{alg1}
\caption{Edge-vertex domination }
\textbf{Input:} An UDG $G=(V, E)$ placed over an hexagonal grid. \\
\textbf{Output:} An {\it EVDS} of $G$.
\begin{algorithmic}[1]
\State Initialize $S^{ev}=\emptyset$, $E'=E$, and let all vertices in $V$ be unmarked initially.
\While {there is an edge $e\in E'$ with its both end-points unmarked}
\State Pick an edge $e\in E'$ such that $e=(u,v)$, the unmarked vertices $u\in A$ and $v\in B$, where $A$ and $B$ are adjacent hexagonal cells
\State Set $S^{ev}=S^{ev}\cup \{e\}$ and $E'=E'\backslash \{e\}$
\State Mark all vertices which are incident to $N_e[e]$
\EndWhile
\While {there is a hexagon containing unmarked vertices}
	\State Pick any arbitrary edge $e\in E$ with at least one endpoint in that hexagon
	\State $S^{ev}=S^{ev}\cup \{e \}$
\EndWhile

\Return $S^{ev}$
%

\end{algorithmic}

\end{algorithm}
\begin{theorem}
Algorithm 1 gives a factor 5-approximation for {\it EVDS} problem on a UDG in $O(m+n)$ time. 
\end{theorem}
\begin{proof}
 Algorithm 1 picks an edge $e$ arbitrarily whose endpoints lie in different hexagons and then repeatedly selects an edge between an unmarked vertex and another vertex in the different hexagon until there are no unmarked vertices (see Fig. \ref{figa2}). In Figure \ref{figa2}, we can observe that Algorithm 1 selected {\it EVDS} as $\{ e_1,e_2,e_3,e_4, e_5 \}$ whose cardinality is five whereas the optimal solution may have a single edge that will ev-dominate every given point (see the edge $e$ in Fig. \ref{figa2}). Next, one can see that the algorithm may select at most five times the optimal value, since an edge between points in two adjacent hexagons may ev-dominate the points in all of its adjacent eight hexagons. As we look at every edge between points to know whether its endpoints are the marked vertices and select an edge at lines 3 and 8 of Algorithm 1, the running time is polynomial in $m$ and $n$.

The approximation factor five of Algorithm 1 follows due to the following two facts:
\begin{enumerate}
\item If both the endpoints of an edge $e$ selected by Algorithm 1 lie within the same hexagon, then none of the vertices corresponding to these points are adjacent to a vertex of its adjacent hexagons.
\item Otherwise an edge $e$ selected by Algorithm 1 ev-dominates all the points in both the adjacent hexagons (Lemma \ref{lema1}).
\end{enumerate}
All the cells (hexagons) in ${\cal R}$ can be grouped as a collection of mega-cells (as in Fig. \ref{figa3}), where each mega-cell consists of ten adjacent hexagonal cells (cells colored with the same color in Fig. \ref{figa3}). Algorithm 1 picks at most five edges to ev-dominate all the points in each mega-cell, whereas in optimal solution at least one edge is required.
\end{proof}

\begin{figure}
\centering
\includegraphics[scale=1]{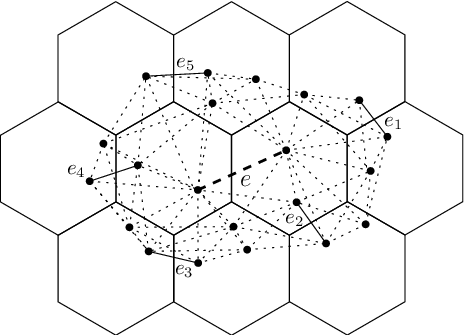}
\caption{ $EVDS =\{e_1,e_2,e_3,e_4,e_5\}$, and minimum $EVDS=\{e\}$}
 \label{figa2}
\centering
\end{figure}

\begin{figure}[htb]
\centering
\includegraphics[scale=0.42]{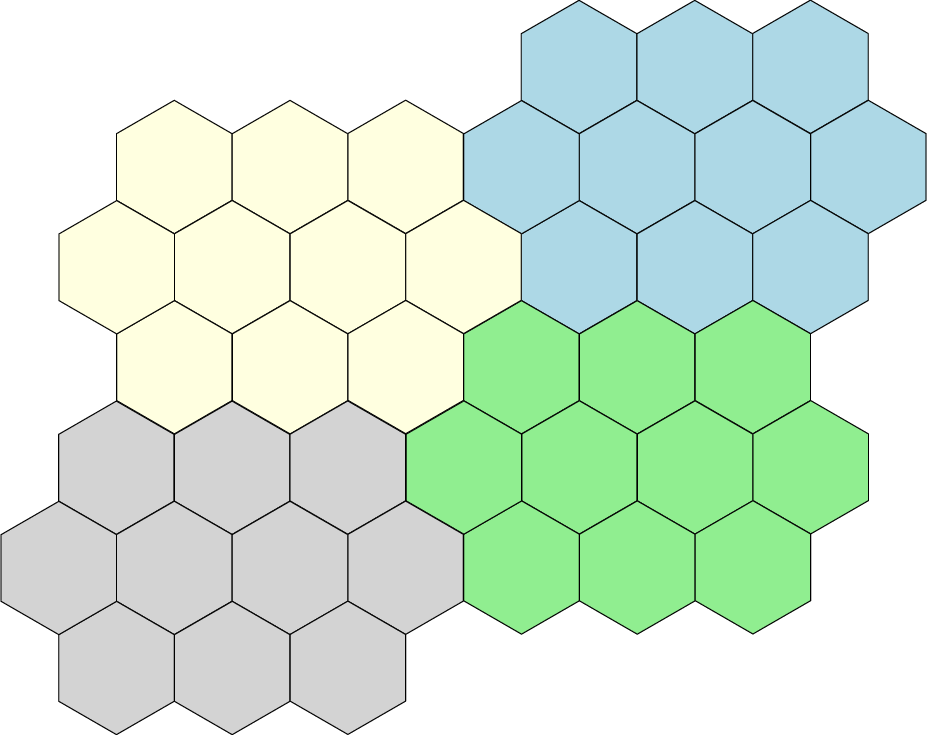}
\caption{Four adjacent Mega-cells}
 \label{figa3}
\centering
\end{figure}



\section{Conclusion}

 In this paper, we have studied the complexity and approximability of the {\it edge-vertex dominating set} problem on unit disk graphs ({\it EVDS-UDG}). We first proved that the decision version of the {\it EVDS-UDG} is {\tt NP-complete}. We then showed that the {\it EVDS-UDG} admits a PTAS. We also gave a simple 5-factor approximation algorithm in linear time. Although this 5-factor approximation algorithm is significantly faster when compared to PTAS, it requires a geometric representation of the input graph, whereas the proposed PTAS does not, hence is robust.

\bibliography{cas-refs}


\end{document}